\documentclass[11pt, twocolumn, copyright, gdm]{google}

\usepackage[authoryear, sort&compress, round]{natbib}

\uselogo{}

\usepackage{natbib}
\usepackage{amsmath, amssymb, amsfonts}
\usepackage{algorithmic}
\usepackage[utf8]{inputenc}
\usepackage[T1]{fontenc}
\usepackage{booktabs}
\usepackage{amsfonts}
\usepackage{nicefrac}
\usepackage{microtype}
\usepackage{enumitem}
\usepackage{graphicx}
\usepackage{gensymb}
\usepackage{caption}
\usepackage{csquotes}
\usepackage{subcaption}
\usepackage{amsthm}
\usepackage{mathrsfs}
\usepackage{amsmath}
\usepackage{amssymb}
\usepackage{mathtools}
\usepackage{soul}
\usepackage[textwidth=1.2in,textsize=tiny]{todonotes}
\usepackage{diagbox}
\usepackage{algorithm}
\usepackage{hhline}
\usepackage{multirow}
\usepackage{adjustbox}
\usepackage{textcomp}
\usepackage{arydshln}
\usepackage{tikz}
\usetikzlibrary{arrows.meta, positioning, shapes.geometric, fit, backgrounds}

\newcommand{\abs}[1]{\lvert #1\rvert}

\newcommand{\bR}{\mathbb{R}}
\newcommand{\norm}[1]{\left\lVert#1\right\rVert}
\newcommand{\unitnorm}[1]{\operatorname{unit}\!\left(#1\right)}

\newcommand{\cL}{\mathcal{L}}

\newcommand{\cT}{\mathcal{T}}
\newcommand{\kl}{\text{KL}}

\newcommand{\xv}{\boldsymbol{x}}

\newcommand{\zv}{\boldsymbol{z}}
\newcommand{\uv}{\boldsymbol{u}}
\renewcommand{\vv}{\boldsymbol{v}}

\newcommand{\bbE}{\mathbb{E}}
\newcommand{\cD}{\mathcal{D}}
\newcommand{\cC}{\mathcal{C}}
\newcommand{\cF}{\mathcal{F}}

\newtheorem{proposition}{Proposition}

\title{ Test-Time Optimization of Query Embeddings\\with Ranking Aware Reward Maximization}

\correspondingauthor{jxwu@google.com\\
This work was done during Tianyu Chen's internship at Google DeepMind.}

\reportnumber{ } 

\renewcommand{\today}{2026-08-11}

\author[1]{Tianyu Chen}
\author[2]{Jiaxing Wu}

\affil[1]{The University of Texas at Austin}
\affil[2]{\thepa{}{}}

\begin{abstract}
Dense retrievers rank documents using vector similarity between a frozen encoder and a precomputed index. While test-time ranking rewards from a reranker or LLM judge can improve results, existing methods discard this signal after a single query. Updating the retriever's weights makes rewards reusable, but this requires parameter access—which is unavailable for closed-source models—and is computationally prohibitive. We propose \textsc{TTT-Embed} (Test-Time Tuning of Embeddings), a framework that distills ranking rewards into a lightweight, learned vector within the output embedding space of a frozen model. This vector is optimized purely from scalar ranking scores assigned to the retriever's own candidate documents, requiring no access to model weights, ground-truth labels, or modifications to index. A single scope parameter controls rewards reuse (global, task, or query), enabling a principled trade-off between reusability and specificity under a fixed reward computation budget. We demonstrate that as the available reward budget scales, the optimal sharing scope shifts dynamically from global-wise to task-wise and finally to query-wise. Evaluated across five embedding models and 15 MTEB retrieval tasks, \textsc{TTT-Embed} improves test-time retrieval by up to +8.36 nDCG@10. Crucially, the learned states generalize effectively to unseen queries (up to +8.57 nDCG@10) and unseen tasks (up to +4.71 nDCG@10). Furthermore, \textsc{TTT-Embed} successfully resolves catastrophic forgetting: by leaving base weights entirely frozen, it recovers degraded general capabilities (up to +8.00 nDCG@10, even surpassing the original base model) while preserving in-domain specialization. These results establish ranking rewards as a reusable test-time state, enabling budget-efficient adaptation for any embedding model, including closed-source APIs.
\end{abstract}

\begin{document}

\maketitle

\section{Introduction}
\begin{figure}[t]
    \centering
    \includegraphics[width=\columnwidth]{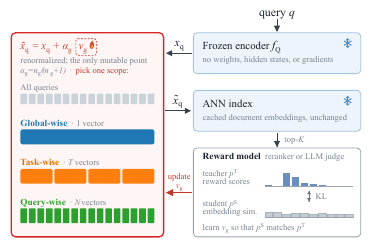}
    \caption{\textbf{\textsc{TTT-Embed} keeps ranking rewards as a reusable
    learned vector in the embedding space.} A query is encoded once by
    the frozen encoder, the learned vector $\vv_g$ is added to the returned query
    vector. The
    vector $\vv_g$ is the only updated quantity. An
    external reward model scores the retrieved candidates; the induced teacher
    distribution $p^{\mathrm{T}}$ is the target that $\vv_g$ is trained to
    match with the student distribution $p^{\mathrm{S}}$ read off the
    embeddings. The sharing scope $g$ determines how broadly each $\vv_g$ is
    reused, and is
    the only degree of freedom that changes between our three variants. Nothing inside the embedding system is modified.}
    \label{fig:teaser}
\end{figure}

Dense retrievers encode queries and documents independently, ranking them by vector similarity \citep{karpukhin2020dense, xiong2020approximate}. Their efficiency stems from a frozen encoder and a precomputed document index. While a reranker provides a stronger relevance signal at test time \citep{nogueira2019passage}, this signal is typically consumed only once. Because direct reranking and query-specific optimization update only the current result list or representation, the ranking rewards is discarded after the query is served.

Model updates offer the opposite trade-off: they preserve information across queries but require access to the retriever's parameters. This is undesirable when deployments rely on cached embeddings, and impossible for closed-source APIs. Moreover, recent embedding models inherit the immense scale of foundation models, such as the 0.6B--8B parameter Qwen3-Embedding or the Gemini family \citep{zhang2025qwen3, lee2025gemini,shanbhogue2026gemini,vera2025embeddinggemma}. Consequently, updating these systems is computationally prohibitive—even when weights are available—and necessitates curating additional training data.

Test-time optimization offers a complementary solution: spending additional computation on deployed inputs while keeping the base model frozen. This paradigm—already successful in distribution shift adaptation and scaling LLM inference \citep{sun2020test, liu2021ttt++, snell2024scaling}—achieves performance gains without requiring model weights or new training data. 

Utilizing ranking rewards from rerankers during test time holds significant promise for retrieval tasks. However, computing a cross-encoder or LLM judge for every query--document pair introduces prohibitive latency and computational costs \citep{hui2022ed2lm}. In practice, real-world systems operate under finite budgets, selecting only a subset of queries for relevance annotation or using cost-aware API cascades \citep{wang2024simple,chen2023frugalgpt}. We provide additional discussion of related work in the Appendix. To address this deployment bottleneck, we ask a central question: \emph{\textbf{How can we best utilize test-time ranking rewards under a fixed computational budget by sharing signals across queries, without accessing or updating model weights?}}

To answer this, we propose \textsc{TTT-Embed} (\emph{Test-Time Tuning of Embeddings}), a framework that distills ranking rewards into a lightweight, reusable state in the output embedding space. Given an initial embedding $\xv_q$ for query $q$, we learn a residual vector $\vv_g$ and retrieve with $\widetilde{\xv}_q=\xv_q+\alpha_g\vv_g$ (Figure~\ref{fig:teaser}). The scope $g$ determines which queries share this learned rewards. We investigate three sharing scopes: global-wise (across all queries), task-wise (restricted to a specific task), and query-wise (private to one query). Because each $\vv_g$ is learned using scalar scores from a reward model applied solely to the query representation, the encoder and ANN index remain entirely frozen. The resulting test-time $\vv_g$ is inexpensive to store and fully compatible with closed-source APIs.

Through extensive experiments across five embedding models on the MTEB benchmark (comprising 15 retrieval tasks) and SKILLRET, we address the following research questions:

\begin{itemize}
    \item \textbf{RQ1:} Does \textsc{TTT-Embed} improve retrieval performance, particularly its ability to generalize to unseen queries and zero-rewards tasks?
    \item \textbf{RQ2:} To what extent can \textsc{TTT-Embed} recover degraded general performance while preserving the in-domain specialization of a fine-tuned model?
    \item \textbf{RQ3:} What are the optimal strategies for real-world deployment regarding rewards budgets, sharing scopes, the breadth-versus-depth trade-off, and tuning-free magnitude computation?
\end{itemize}

Our contributions are threefold:

\begin{enumerate}[leftmargin=1.5em, itemsep=3pt]
    \item \textbf{A Novel Test-Time Tuning Framework:} We propose \textsc{TTT-Embed}, which introduces a theoretically grounded, vector-based adaptation interface ($\xv_q+\alpha_g\vv_g$) with flexible sharing scopes. Operating entirely on output embeddings and off-the-shelf reward feedback, it requires no access to model weights or ground-truth labels, making it suited for both open-source and API-based models.
    \item \textbf{Generalization and Forgetting Recovery:} We empirically demonstrate that \textsc{TTT-Embed} yields significant gains on standard benchmarks and generalizes effectively to unseen queries and tasks. Furthermore, by keeping model weights entirely frozen, it resolves the trade-off between specialization and generalizability—recovering the degraded general capabilities of fine-tuned models without sacrificing their in-domain expertise.
    \item \textbf{Deployment Optimization Guidelines:} We provide a thorough analysis of how fixed computational budgets dictate deployment strategies, offering concrete guidelines for selecting the optimal sharing scope, navigating the breadth-versus-depth trade-off, and automatically determining the application magnitude ($\alpha_g$) to eliminate manual tuning.
\end{enumerate}
\section{\textsc{TTT-Embed}: Reusable Test-Time Tuning in Embedding Space}
\label{sec:method}

\textsc{TTT-Embed} converts scalar ranking rewards into a learned query-side
vector outside a frozen embedding model. Its learning rule is unchanged
across sharing scopes; only which rewards queries learn one $\vv_g$ and
which served queries reuse it differ.

\subsection{Problem Setup and Reward Budget}

Let a potentially closed-source embedding system provide a query encoder
$f_{\mathrm{Q}}$ and document encoder $f_{\mathrm{D}}$, which may be the same underlying model. Then for any nonzero vector $\boldsymbol{a}$, let $\unitnorm{\boldsymbol{a}}\coloneqq \boldsymbol{a}/\norm{\boldsymbol{a}}$ denote its $\ell_2$ normalization. For a
query $q$ and document $d$, we write
$\xv_q=\unitnorm{f_{\mathrm{Q}}(q)}$ and
$\zv_d=\unitnorm{f_{\mathrm{D}}(d)}$, with base retrieval score
$s(q,d)=\xv_q^\top\zv_d$. The document embeddings
$\{\zv_d\}_{d\in\cD}$ are computed once and stored in an approximate
nearest-neighbor (ANN) index. We require only vector outputs from the
embedding system; its parameters, hidden states, and gradients are never
accessed.

Let $E$ be a workload of $N=\abs{E}$ queries and let
$\cF\subseteq E$ contain the queries selected for rewards. For each
$q\in\cF$, the base retriever returns top-K documents
$\cC_q=\{d_{q1},\ldots,d_{qk_q}\}$, and an external reward model---such as a
cross-encoder reranker or an LLM judge---assigns each pair a scalar relevance
score $r_{qi}$. A larger value denotes greater predicted relevance. We count
one scored query--document pair as one atomic reward, giving total and
normalized budgets
\begin{equation}
    B = \sum_{q\in\cF}\abs{\cC_q},
    \qquad
    b = \frac{B}{N}.
    \label{eq:budget}
\end{equation}
Thus, $b$ is the average number of external rewards available per 
query, even when they are not allocated uniformly. The budget excludes the
base embedding and ANN-search costs shared by all systems. Ground-truth
relevance labels are used only for evaluation, never to select rewards or
learn $\vv_g$.

\subsection{Learning $\vv_g$ from Ranking Rewards}

For now, let $g$ denote any group of rewarded queries
$\cF_g\subseteq\cF$ whose rewards are used to learn one vector
$\vv_g\in\bR^p$; the next subsection instantiates the possible groups. A query
assigned to $g$ is represented as
$\widetilde{\xv}_q(\vv_g)=\unitnorm{\xv_q+\alpha_g\vv_g}$, where
$\alpha_g\in[0,1]$ controls how strongly $\vv_g$ is applied.

\begin{algorithm}[t]
\caption{Learning and Deploying Scoped Vectors}
\label{alg:scoped_vector}
\begin{algorithmic}[1]
\REQUIRE Frozen embeddings $\{\xv_q\}$ and $\{\zv_d\}$; scoped rewards
groups $\{\cF_g\}$ with candidates $\cC_q$ and rewards $r_{qi}$
\ENSURE Stored test-time state $\{(\vv_g^*,\alpha_g)\}$
\FOR{each rewards group $g$ induced by the chosen scope}
    \STATE Form $p_q^{\mathrm T}=\operatorname{softmax}(r_q/\tau_{\mathrm T})$
    for every $q\in\cF_g$
    \STATE Initialize $\vv_g\leftarrow\mathbf{0}$ and optimize
    Equation~\eqref{eq:vector_objective}
    \STATE Set $n_g\leftarrow|\cF_g|$ and $\alpha_g\leftarrow n_g/(n_g+1)$
    \STATE Store $(\vv_g^*,\alpha_g)$ outside the embedding model
\ENDFOR
\FOR{each served query $q$ with an applicable group $g$}
    \STATE $\widetilde{\xv}_q\leftarrow
    \unitnorm{\xv_q+\alpha_g\vv_g^*}$
    \STATE Retrieve with $\widetilde{\xv}_q$ from the unchanged ANN index
\ENDFOR
\STATE Queries without an applicable state use their base embedding $\xv_q$
\end{algorithmic}
\end{algorithm}

This additive form preserves the retrieval interface. For a fixed query,
normalization is a positive constant across documents, so
$\widetilde{\xv}_q(\vv_g)^\top\zv_d$ ranks them exactly as the base score
$\xv_q^\top\zv_d$ plus the additive term
$\alpha_g\vv_g^\top\zv_d$. At deployment, we use
$\widetilde{\xv}_q$ as the query vector for a new search over the unchanged
ANN index. It has the same dimensionality and uses the same similarity
function as $\xv_q$, so neither the indexed document embeddings nor the index
structure needs to be modified. This second retrieval may retrieve documents
outside the original reward set $\cC_q$.
Appendix proves the ranking equivalence and
states its implementation details.

We distill the reward model's scores over each candidate set. First, we form the
teacher distribution
\begin{equation}
    p^{\mathrm{T}}_{qi}
    =
    \frac{\exp(r_{qi}/\tau_{\mathrm{T}})}
         {\sum_{j=1}^{k_q}\exp(r_{qj}/\tau_{\mathrm{T}})}.
    \label{eq:teacher_distribution}
\end{equation}
For a candidate vector $\vv$, we define the induced student distribution as $\overline{\xv}_q(\vv)
    = \unitnorm{\xv_q+\vv}$
\begin{equation}
    p^{\mathrm{S}}_{qi}(\vv)
    =
    \frac{\exp\!\left(\overline{\xv}_q(\vv)^\top\zv_{d_{qi}}/\tau_{\mathrm{S}}\right)}
         {\sum_{j=1}^{k_q}\exp\!\left(\overline{\xv}_q(\vv)^\top\zv_{d_{qj}}/\tau_{\mathrm{S}}\right)}.
    \label{eq:student_distribution}
\end{equation}
We learn $\vv_g$ by knowledge distillation with ridge regularization:
\begin{equation}
    \vv_g^*
    =
    \arg\min_{\vv\in\bR^p}
    \sum_{q\in\cF_g} w_q\,
    \kl\!\left(p^{\mathrm{T}}_q\,\middle\|\,p^{\mathrm{S}}_q(\vv)\right)
    + \lambda\norm{\vv}^2
    \label{eq:vector_objective}
\end{equation}
where $\sum_{q\in\cF_g}w_q=1$. Equation~\eqref{eq:vector_objective} is a listwise soft-label knowledge
distillation objective \citep{hinton2015distilling,cao2007learning}: it retains
the reward model's relative scores within a candidate set rather than
constructing binary positive--negative pairs or updating the encoder with an
in-batch contrastive loss. Our contribution lies in the learned vector and
its reuse across scopes.
Analysis section compares alternative
objectives, and theoretical analysis of the objective and definition of $w_g$ are
deferred to the appendix.

\subsection{Sharing Scopes for Learned Vectors}

The group assignment determines the persistence of the learned state:
\begin{itemize}[leftmargin=1.5em, itemsep=2pt]
    \item \textbf{Global-wise:} one $\vv_g$ pools rewards across tasks and is
    reused by every query served by the same embedding model.
    \item \textbf{Task-wise:} one $\vv_g$ pools and reuses rewards within
    each retrieval task.
    \item \textbf{Query-wise:} one $\vv_g$ is learned for each rewarded query
    and remains private to that query.
\end{itemize}
For a global-wise state, the weights $w_q$ in
Equation~\eqref{eq:vector_objective} first average queries within each task and then average tasks uniformly. For a task-wise state,
they reduce to the ordinary average over queries; for a query-wise state, they
reduce to a single-query loss. Appendix gives the
unified formula.

When every query receives rewards for $b$ retrieved documents, $\cF=E$ and $B=Nb$, yielding the standard full-coverage setting. All scopes then use the same query--document rewards and differ only in which queries share a common state $\vv_g$. Under a limited budget, $\cF$ may cover only a subset of $E$, and the reward depth $b_q$ may vary across queries. A global $\vv_g$ can be applied to all queries, while a task-wise $\vv_g$ can be applied to queries within the same task. In contrast, a query-wise state affects only queries in $\cF$; queries without rewards retain their base embeddings.

\textsc{TTT-Embed} does not prescribe which queries or documents receive rewards; it accepts any externally selected set of rewarded query--document pairs. For controlled evaluation, we use a task-balanced water-filling allocation with nested budgets, ensuring that every sharing scope and baseline receives the same reward pairs at each budget. The full protocol is provided in the Appendix.

\subsection{Confidence-Aware Magnitude for Deployment}
\label{sec:confidence}

Equation~\eqref{eq:vector_objective} first learns $\vv_g^*$. At deployment, its appropriate magnitude depends on estimation confidence, which increases with the reward budget. Let $n_g$ denote the number of rewarded queries used to learn state $\vv_g$. We apply the evidence-adaptive scaling rule
\begin{equation}
    \alpha_g = \frac{n_g}{n_g+1},
    \label{eq:alpha}
\end{equation}
and deploy $\alpha_g\vv_g^*$. The following proposition provides a Bayesian justification for this label-free rule. Table~\ref{tab:alpha_ablation} confirms its practical benefits: without tuning on evaluation labels, it achieves the lowest average regret (0.11) to the oracle, compared with 0.17 for the unshrunk choice $\alpha=1$.

\begin{proposition}[Evidence-adaptive vector shrinkage]
\label{prop:bayesian_shrinkage}
Let $\uv_g\in\bR^p$ be a latent transferable vector and
$\boldsymbol{\Sigma}\succ0$. If
$\uv_g\sim\mathcal{N}(\mathbf{0},\boldsymbol{\Sigma})$ and
$\vv_g^*\mid\uv_g,n_g\sim
\mathcal{N}(\uv_g,\boldsymbol{\Sigma}/n_g)$ for $n_g\geq1$, then the Bayes
estimator under squared error is
\begin{equation}
    \mathbb{E}[\uv_g\mid\vv_g^*,n_g]
    =\frac{n_g}{n_g+1}\vv_g^*.
    \label{eq:posterior_dose}
\end{equation}
\end{proposition}

\begin{proof}
The prior and likelihood precisions are $\boldsymbol{\Sigma}^{-1}$ and
$n_g\boldsymbol{\Sigma}^{-1}$. Hence the posterior precision is
$(n_g+1)\boldsymbol{\Sigma}^{-1}$ and its information vector is
$n_g\boldsymbol{\Sigma}^{-1}\vv_g^*$, giving
Equation~\eqref{eq:posterior_dose}.
\end{proof}

This normal-means model motivates uncertainty calibration but is not assumed when optimizing Equation~\eqref{eq:vector_objective}. The scaling factor $\alpha_g$ shrinks toward $0.5$ for small $n_g$ and approaches $1$ as more rewarded queries are pooled. Shared states count all rewarded queries within their scope, whereas a query-wise state has $n_g=1$ and thus $\alpha_g=0.5$. We fix this rule a priori and never tune it on evaluation labels.

Algorithm~\ref{alg:scoped_vector} summarizes the complete procedure. The
learning and deployment rules are identical for all three scopes; changing the
scope changes only which rewarded queries share a group $g$ and hence one
stored state $(\vv_g^*,\alpha_g)$.

\section{Experimental Setup}
\label{sec:experiments}
\subsection{Tasks and Embedding Models}
We evaluate our approach on 15 English retrieval tasks from the MTEB benchmark \citep{muennighoff2022mteb}, spanning question answering, fact verification, duplicate detection, argument retrieval, and scientific and financial search. In total, the suite comprises $12{,}263$ test queries (ranging from 43 to 1,595 per task) and $1{,}661{,}208$ corpus entries (ranging from $8{,}674$ to $303{,}732$ per task). To prevent large datasets from skewing the results, we average performance equally across all tasks. Since real-world deployments frequently utilize query classification components tailored to specific business needs, MTEB's predefined task boundaries serve as an appropriate proxy for our task-wise evaluation.

We additionally use SKILLRET \cite{cho2026skillret}, an agent-skill retrieval benchmark comprising $4{,}997$ queries and $6{,}660$ candidate skill cards, as the specialization corpus in the degraded-performance recovery experiment described in the preceding section. SKILLRET was released after the training-data cutoffs of all embedding and reward models used in our experiments, allowing us to evaluate whether \textsc{TTT-Embed} can recover performance degraded by domain-specific fine-tuning while minimizing the risk of benchmark contamination.

We use five embedding models: Qwen3-Embedding-0.6B,
Qwen3-Embedding-4B, EmbeddingGemma-300M, Gemini Embedding 1, and Gemini
Embedding 2 \citep{zhang2025qwen3,vera2025embeddinggemma,lee2025gemini,shanbhogue2026gemini}. The first three
expose their weights, whereas the two Gemini models are accessed only through
their embedding APIs. Document embeddings and the
search index remain fixed. 

\subsection{Metrics}
We report nDCG@10 and its matched change from the raw-query retriever,
$\Delta\mathrm{nDCG@10}$.  To summarize a full budget curve, we also report
the area under the gain--budget curve (AUBC), and details defer to Appendix. 

\subsection{Reward Models}
Our reward model is Qwen3-Reranker-4B \cite{zhang2025qwen3}. For each query--document pair,
we compute its score from the raw forward-pass logits of the \emph{Yes} and
\emph{No} tokens:
\begin{equation}
    r(q,d)
    = P(\mathrm{Yes}\mid q,d)
    = \frac{\exp(\ell_{\mathrm{Yes}})}
           {\exp(\ell_{\mathrm{Yes}})+\exp(\ell_{\mathrm{No}})}
    \in[0,1].
    \label{eq:reranker_reward}
\end{equation}
This computation is deterministic and does not sample text. For the
general-purpose judges, we instead use a fixed 0--5 relevance rubric.
Their final numeric grade is mapped to $[0,1]$.
Appendix specifies the prompts, score extraction,
and generation settings for each reward model.
Section~\ref{sec:reward_model_selection} compares the resulting reward
models on the same candidate pairs.

\subsection{Baselines}

For an evaluation set of $N$ queries, we report the normalized reward budget $b=B/N$, where $B$ is the total number of judgments and $b$ is the average number of rewards per query. We allocate rewards using a task-balanced water-filling schedule with nested budgets, as detailed in the Appendix. All methods receive exactly the same rewarded query--document pairs. We compare the following methods.

\textbf{Raw retrieval} uses no rewards and ranks the corpus according to
$\xv_q^\top\zv_d$.
\textbf{Direct reranking} reorders the rewarded portion of each selected query's retrieved candidate list according to $r(q,d)$.
\textbf{Query-wise \textsc{TTT-Embed}} learns a separate state $\vv_q$ from the rewards associated with each query and reranks its candidates using the adapted query embedding $\xv_q+\alpha\vv_q$.
\textbf{Task-wise and global \textsc{TTT-Embed}} pool rewards at the task or global level to learn a shared state $\vv_g$, which is applied to all queries within the corresponding scope. Under partial query coverage, only the shared variants can adapt rankings for queries that receive no rewards. None of the methods updates the encoder parameters or document embeddings.

\section{Results}
\label{sec:budget_results}

\subsection{Performance Gains and Optimal Sharing Scope Across Reward Budgets}

\begin{figure*}[t]
    \centering
    \includegraphics[width=\textwidth]{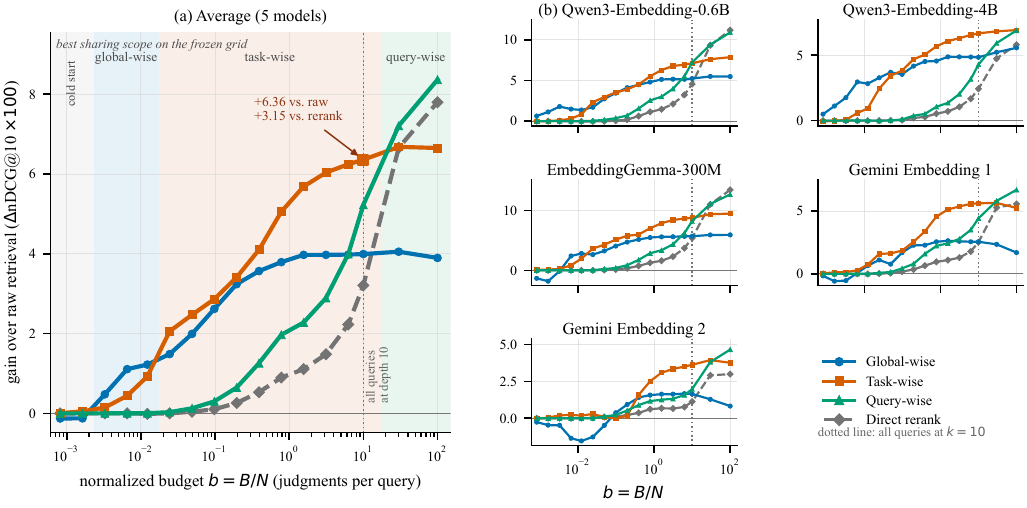}
    \caption{\textbf{\textsc{TTT-Embed} improves retrieval, and the optimal scope changes with the budget.} Panel (a) reports the performance gain over raw retrieval across different normalized budgets, averaged over five embedding models. The background shading indicates the best sharing scope for each sampled budget regime. Panel (b) details these results for each individual embedding model. The dotted line marks full query coverage at depth 10. At this threshold, task-wise \textsc{TTT-Embed} improves raw retrieval by 6.36 nDCG@10 points and outperforms cost-matched direct reranking by 3.15 points. Rewards are from Qwen3-Reranker-4B.}
    \label{fig:budget_scope}
\end{figure*}

\begin{table}[t]
    \centering
    \scriptsize
    \setlength{\tabcolsep}{2.6pt}
    \caption{
    Representative scope-optimal operating points across
    budget regimes. Entries are gains over raw retrieval in
    $\Delta\mathrm{nDCG@10}\times100$. The three columns select the best
    \textsc{TTT-Embed} scope on the average curve at representative low,
    intermediate, and high budgets in Figure~\ref{fig:budget_scope}. Direct
    reranking is evaluated at the same column-specific budget; exact grid construction is given in Appendix~\ref{app:experimental_details}.}
    \label{tab:scope_peak}
    \begin{tabular}{lrrr}
        \toprule
        & \multicolumn{3}{c}{\textsc{TTT-Embed} state scope} \\
        \cmidrule(lr){2-4}
        Embedding model
            & \shortstack{Global}
            & \shortstack{Task}
            & \shortstack{Query} \\
        \midrule
        Budget $b$ & $10^{-2}$ & 10 & 100 \\
        \midrule
        Qwen3-Embedding-0.6B &  1.37 & 7.12 & 10.92 \\
        Qwen3-Embedding-4B   &  2.82 & 6.67 &  6.89 \\
        EmbeddingGemma-300M  &  2.81 & 8.76 & 12.67 \\
        Gemini Embedding 1   &  0.66 & 5.62 &  6.69 \\
        Gemini Embedding 2   & $-1.54$ & 3.63 &  4.66 \\
        \midrule
        \textbf{\textsc{TTT-Embed} Average}
                       & \textbf{1.23} & \textbf{6.36} & \textbf{8.36} \\
        Direct rerank  & $-0.01$ & 3.21 & 7.80 \\
        \textsc{TTT-Embed} $-$ direct
                       & \textbf{$+1.23$} & \textbf{$+3.15$}
                       & \textbf{$+0.56$} \\
        \bottomrule
    \end{tabular}
\end{table}

Figure~\ref{fig:budget_scope} presents the main results of \textsc{TTT-Embed} across 15 MTEB retrieval tasks, three sharing scopes, and a range of normalized reward budgets. Results are averaged over five embedding models: Qwen3-Embedding-0.6B, Qwen3-Embedding-4B, EmbeddingGemma-300M, Gemini Embedding 1, and Gemini Embedding 2. The results reveal a clear progression from broad reuse to query-specific specialization. On our discrete budget grid, global sharing performs best for $b\in[0.003,0.012]$, as a state learned from even limited rewards can be applied to the entire query population. Task-wise sharing performs best from $b=0.024$ through $b=10$, whereas query-wise optimization becomes optimal at $b=30$ and remains so at $b=100$. While these operating points should not be interpreted as universal thresholds. The model-specific panels also reveal differences in different embedding models. Nevertheless, the upper envelope across the three sharing scopes remains above direct reranking throughout the average budget curve. When integrated over the full log-budget range, task-wise sharing also achieves the highest gain AUBC@100 for all five embedding models (see Appendix).

At $b=10$, corresponding to an average of ten judgments per query, task-wise \textsc{TTT-Embed} improves nDCG@10 over raw retrieval by 6.36 points on average (Figure~\ref{fig:budget_scope}). The improvement is positive for every embedding model, ranging from 3.63 to 8.76 points. Using the same rewarded pairs, direct reranking yields an average improvement of only 3.21 points, giving the learned task-wise state a 3.15-point advantage. At $b=100$, query-wise optimization achieves an 8.36-point gain and continues to outperform direct reranking by 0.56 points. These results show that the benefit does not arise solely from applying reward-model scores to retrieved candidates. Distilling these scores into the query embeddings enables corpus-wide retrieval improvements beyond those obtained by using the scores only to reorder the original candidate list.

Table~\ref{tab:scope_peak} highlights three representative points from the scope transition shown in Figure~\ref{fig:budget_scope}. At $b=10^{-2}$, global sharing is the best-performing scope on the average curve, improving nDCG@10 by 1.23 points, compared with a $-0.01$-point change from direct reranking. At $b=10$, task-wise sharing yields a 6.36-point gain and outperforms direct reranking by 3.15 points. At $b=100$, query-wise optimization achieves an 8.36-point gain while retaining a 0.56-point advantage over direct reranking. Thus, the best sharing scope for each budget regime consistently outperforms cost-matched direct reranking at all three representative operating points.



\subsection{Generalization to Unseen Queries and Tasks}
\label{sec:transfer_results}

\begin{table}[t]
    \centering
    \scriptsize
    \setlength{\tabcolsep}{2.4pt}
\caption{\textbf{The learned $\vv_g$ improves retrieval on held-out queries and tasks.}
Entries report gains over raw retrieval in $\Delta\mathrm{nDCG@10}\times100$. The first two columns evaluate global and task-wise sharing when 80\% of the queries each provide 10 rewards, with performance measured on the remaining 20\% of unseen queries. The final column reports the leave-one-task-out (LOTO) gain achieved by a global $\vv_g$ learned from the other 14 tasks.}
    \label{tab:transfer}
    \begin{tabular}{@{}lrrr@{}}
        \toprule
        & \multicolumn{2}{c}{Held-out query}
        & \multicolumn{1}{c}{Held-out task} \\
        \cmidrule(lr){2-3}\cmidrule(l){4-4}
        Embedding model
            & \shortstack{Global}
            & \shortstack{Task}
            & LOTO \\
        \midrule
        Qwen3-Embedding-0.6B & 4.88 & 6.03 & 4.36 \\
        Qwen3-Embedding-4B   & 4.93 & 6.14 & 4.16 \\
        EmbeddingGemma-300M  & 6.42 & 8.57 & 4.71 \\
        Gemini Embedding 1   & 2.40 & 3.95 & 1.36 \\
        Gemini Embedding 2   & 1.49 & 3.17 & 0.85 \\
        \midrule
        Average              & 4.02 & 5.57 & 3.09 \\
        \bottomrule
    \end{tabular}
\end{table}

In the previous section, at budget $B$, we denote the queries with rewards by $S_B$ and the remaining queries by $U_B=E\setminus S_B$. We learn $\vv_g$ at each sharing scope using rewards from $S_B$ and evaluate it on the full query set $E$. The fixed-population curves establish strong utility under a standard deployment budget. To ensure this performance translates smoothly to real-world deployments—where the system must handle entirely unseen queries, we conduct two stricter transfer experiments to evaluate generalization to unseen queries and tasks. For \emph{global and task-wise transfer}, we learn a shared $\vv_g$ using rewards from $S_B$ and evaluate it exclusively on the unseen queries $U_B$. Because queries in $U_B$ contribute no rewards during optimization, any improvement on this subset must arise from information shared within tasks or globally. For \emph{cross-task transfer}, we learn a global state $\vv_g$ from 14 tasks and apply it to the remaining task without using any rewards from the target task. We repeat this procedure with each of the 15 tasks serving as the held-out task.

Table~\ref{tab:transfer} shows positive held-out-query endpoint gains for
both shared scopes and all five embedding models. When 80\% queries has 10 rewards coverage and remaining 20\% queries are unseen,
the task-wise state improves the untouched queries by 3.17--8.57 points
(5.57 on average), despite those queries never being scored; global-wise sharing
also gains 4.02 points on average. The stronger leave-one-task-out (LOTO)
test is positive on average for every model, including both closed Gemini
models: a state learned on the other 14 tasks improves the held-out task by
0.85--4.71 points, with a 3.09-point average gain. These results directly
establish that the learned state stores reusable information rather than
merely fitting the queries that produced it. Full
transfer curves and per-task results appear in
Appendix.

\subsection{General Capability Gain While Preserving Domain Specialization}
\label{sec:sft_recovery}

\begin{figure}[t]
    \centering
    \includegraphics[width=\columnwidth]{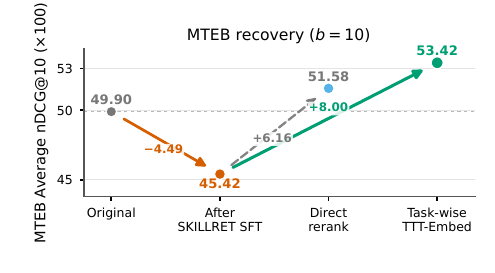}
    \caption{\textbf{\textsc{TTT-Embed} recovers performance lost during domain specialization without weight updates.}
    Fine-tuning Qwen3-Embedding-0.6B on SKILLRET improves in-domain nDCG@10 from 52.24 to 66.21 but reduces the 15-task MTEB average from 49.90 to 45.42. With $b=10$, task-wise \textsc{TTT-Embed} raises the frozen model's MTEB score to 53.42, exceeding both the base model (49.90) and direct reranking (51.58), while preserving the specialized model weights. Scores are nDCG@10 $\times100$.}
    \label{fig:sft_recovery}
\end{figure}

Although domain-specific fine-tuning can reliably improve in-domain performance, it often degrades performance on broader benchmarks, creating a trade-off between specialization and generalization. \textsc{TTT-Embed} provides a elegant solution to this trade-off: it can recover lost capabilities at test time without modifying the underlying model parameters.

In Figure~\ref{fig:sft_recovery}, we demonstrate this utilizing the SKILLRET task \cite{cho2026skillret}. Specifically, we evaluate SKILLRET-SFT-0.6B, a model fine-tuned on top of the base Qwen3-Embedding-0.6B. While this specialization yields target-domain improvements, it suffers out-of-domain generalization, showing a 4.49-point degradation across the 15 MTEB tasks. \textsc{TTT-Embed} can not only recover this performance degradation, but also surpass the original, unspecialized model. It yields an 8.00-point absolute improvement that elevates the MTEB score to 53.42—exceeding the original base model by 3.51 points and outperforming direct reranking by 1.84 points. Crucially, \textsc{TTT-Embed} achieves these \textbf{generalization gains while fully maintaining in-domain specialization}, since the underlying model weights are entirely frozen. Comprehensive global-wise, task-wise, and query-wise results, alongside direct reranking baselines, are detailed in Appendix.

\section{Analysis}
\subsection{Breadth vs Depth for a Fixed  Reward Budget}
\label{sec:breadth_depth}
For a fixed reward budget $b$, allocation within a given sharing scope involves a trade-off between two strategies: a \textit{breadth-primary} approach, which distributes rewards across a larger number of queries with fewer judged documents per query, and a \textit{depth-primary} approach, which concentrates rewards on more documents per query across a smaller subset of queries. Formally, let $\rho=m/N$ denote the fraction of queries receiving rewards (characterizing breadth), and let $k$ denote the number of documents judged per selected query (characterizing depth). The total budget allocation is therefore defined as $b=\rho k$.

\begin{figure}[t]
    \centering
    \includegraphics[width=\columnwidth]{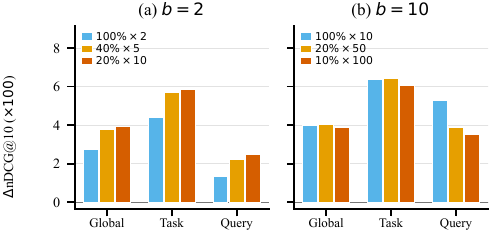}
    \caption{\textbf{Reward allocation at a fixed budget.}
    At the smaller budget $b=2$, panel (a) shows that rewards depth is more valuable than maximal query coverage. At the bigger budget $b=10$, panel (b) shows that broader coverage becomes more beneficial. Results are averaged gains over raw retrieval in $\Delta\mathrm{nDCG@10}\times100$.}
    \label{fig:breadth_depth}
\end{figure}

Figure~\ref{fig:breadth_depth} demonstrates that the optimal balance between breadth and depth is budget-dependent. Under a small budget ($b=2$, Panel a), a depth-primary allocation yields superior performance across all three scopes, suggesting that rewards must encompass a critical mass of documents per query to provide reliable ranking signals. Conversely, with a larger budget ($b=10$, Panel b), a breadth-primary allocation proves more effective. This highlights diminishing returns for depth: once a sufficient number of pairs have provided rewards, expanding query coverage yields greater performance gains than gathering deeper rewards.




\subsection{Tuning-Free Evidence-Adaptive Magnitude}
\label{sec:alpha_ablation}
We evaluate whether the evidence-adaptive magnitude in Eq~\eqref{eq:alpha} performs competitively with manually tuned alternatives. Specifically, we compare $\alpha(n)=n/(n+1)$ with seven fixed values ranging from $0.1$ to $1$. For each task, we learn a task-wise state $\vv_g$ from $n\in\{1,16,128,\mathrm{all}\}$ rewarded queries, each with 10 rewards, and evaluate it on the full query set. As shown in Table~\ref{tab:alpha_ablation}, the evidence-adaptive $\alpha$ achieves the best average performance across the four reward budgets, avoiding the need for budget-specific tuning.

\begin{table}[t]
    \centering
    \scriptsize
    \setlength{\tabcolsep}{2.6pt}
    \caption{\textbf{Evidence-adaptive dosing avoids per-budget label
    tuning.} Entries are task-wise $\Delta\mathrm{nDCG@10}\times100$,
    averaged over five embedding models and 15 tasks. Oracle regret is the average
    gap over the four displayed evidence levels to the best dose in the
    complete sweep at each level (lower is better). Bold and underlining mark
    the best and second-best reported values in each column, respectively,
    determined before rounding; tied values share the same rank.}
    \label{tab:alpha_ablation}
    \begin{tabular}{lrrrrr}
        \toprule
        $\alpha$ & $n{=}1$ & $n{=}16$ & $n{=}128$ & $n{=}\mathrm{all}$
                    & Oracle regret $\downarrow$ \\
        \midrule
         $\alpha=0.1$ & 0.30 & 0.96 & 1.12 & 1.17 & 3.30 \\
         $\alpha=0.2$ & 0.53 & 1.73 & 2.14 & 2.21 & 2.53 \\
         $\alpha=0.5$ & \textbf{0.93} & \textbf{3.43}
                           & 4.43 & 4.70 & 0.81 \\
         $\alpha=1$   & \underline{0.84} & 3.22 & \underline{5.59}
                           & \textbf{6.39} & \underline{0.17} \\
        \hdashline
         $n/(n+1)$ & \textbf{0.93} & \underline{3.40} & \textbf{5.59}
                           & \underline{6.38} & \textbf{0.11} \\
        \bottomrule
    \end{tabular}
\end{table}


\subsection{Reward Model Selection}
\label{sec:reward_model_selection}

We compare reward models by using each one to score and reorder the same
frozen top-10 candidate sets. Qwen3-Reranker-4B supplies the normalized
probability of \emph{Yes} computed from its \emph{Yes}/\emph{No} logits
(Equation~\eqref{eq:reranker_reward}). Each general-purpose LLM instead
returns a 0--5 relevance score for the same query--document pair. Appendix provides the details.

    \begin{table}[t]
    \centering
    \scriptsize
    \setlength{\tabcolsep}{5pt}
    \caption{\textbf{In-domain reward-model selection.}
    Every configuration scores and reorders the same 307,319 MTEB
    query--document pairs. Scores and changes from raw retrieval are
    average nDCG@10 $\times100$ over five embedding models and 15 tasks.}
    \label{tab:reward_model_quality}
    \begin{tabular}{@{}lrr@{}}
        \toprule
        Reward model & Average & $\Delta$ raw \\
        \midrule
        None (raw retrieval) & 53.84 & --- \\
        \midrule
        Qwen3-Reranker-4B & 57.09 & +3.25 \\
        Gemini 3.1 Pro & 56.57 & +2.73 \\
        Gemini 3.5 Flash & 55.94 & +2.10 \\
        \midrule
        {\shortstack[l]{Qwen3-VL-30B-A3B-Instruct}}
            & 50.52 & $-3.32$ \\
        Qwen3-8B-Instruct & 48.97 & $-4.87$ \\
        \bottomrule
    \end{tabular}
\end{table}

Table~\ref{tab:reward_model_quality} shows that reward-model specialization matters. The retrieval-specific Qwen3-Reranker-4B performs best, improving raw retrieval by 3.25 points and outperforming the strongest general-purpose judge by 0.52 points. Nevertheless, capable general-purpose LLMs can provide useful ranking rewards: Gemini 3.1 Pro and Gemini 3.5 Flash improve raw retrieval by 2.73 and 2.10 points, respectively. They therefore provide practical alternatives when a specialized reranker is unavailable. In contrast, the two general-purpose Qwen models underperform raw retrieval, demonstrating that not every off-the-shelf LLM is an effective reward model.


\subsection{Compatibility with Different Learning Objectives}
\label{sec:loss_ablation}

\textsc{TTT-Embed} is compatible with other test-time learning objectives. We compare the knowledge-distillation objective in Equation~\eqref{eq:vector_objective} with a reverse-KL variant and a GRPO-based objective. Softmax-GRPO adapts Retrieval-GRPO \citep{liu2025taosearchemb}: the ten candidate documents form a group, their embedding similarities define the policy, and standardized reward-model scores serve as document-level advantages.

All objectives use the same Reranker-4B rewards, candidate sets, zero initialization, ridge regularization, and 300 Adam steps. Each of the $N$ queries receives $b=10$ reward scores. Table~\ref{tab:loss_ablation} compares the objectives across three sharing scopes and five embedding models at this budget, while Appendix reports the complete nested-budget sweep.

\begin{table}[t]
    \centering
    \small
    \caption{\textbf{Objective ablation at $b=10$ (full depth-10 coverage).}
    Entries are $\Delta\mathrm{nDCG@10}\times100$ averaged over five
    embedding models. Knowledge distillation leads at both shared scopes.}
    \label{tab:loss_ablation}

    \resizebox{0.7\columnwidth}{!}{%
    \begin{tabular}{@{}lrrr@{}}
        \toprule
        Objective & Global & Task & Query \\
        \midrule
        Knowledge Distillation & \textbf{4.00} & \textbf{6.38} & 5.26 \\
        Reverse KL variant & 3.72 & 4.66 & \textbf{5.30} \\
        Softmax-GRPO & 2.07 & 3.13 & 3.27 \\
        \bottomrule
    \end{tabular}%
    }
\end{table}

As shown in Table~\ref{tab:loss_ablation}, knowledge distillation achieves the strongest performance on average. The other objectives also yield positive gains, demonstrating that \textsc{TTT-Embed} is compatible with and robust across different test-time learning objectives. We therefore use knowledge distillation as the default objective for the shared-state setting central to this work. The appendix formalizes its relationship to reverse KL and Softmax-GRPO and presents the details.

\section{Conclusions}
We presented \textsc{TTT-Embed}, a framework that treats ranking rewards as a reusable test-time state by distilling it into a lightweight vector added to frozen query embeddings. This approach requires no access to model weights or modifications to document indexes. We demonstrate that vectors learned from a subset of queries not only improve overall performance and generalize to unseen queries and tasks, but also successfully recover from catastrophic forgetting during test time. Furthermore, the reward budget governs the optimal sharing scope. Future work will extend this framework to multimodal
domains, develop more strategic scope-selection mechanisms for deployment,
and introduce dynamic adaptation based on query complexity.

\section{Acknowledgment}

We extend our sincere gratitude to Mingyuan Zhou, Shanfeng Zhang, Mojtaba Seyedhosseini, Howard Zhou, and Tom Duerig for providing essential resources, alongside their invaluable guidance and support throughout this project. Additionally, we thank Tanmaya Dabral for reviewing the manuscript and offering insightful technical feedback.

\bibliography{aaai2027}

\clearpage
\newpage

\section{Related Work}
\label{app:related_work}

\paragraph{Dense retrieval and reranking.}
Bi-encoder retrievers embed queries and documents independently, enabling document representations to be indexed before queries arrive \citep{karpukhin2020dense, reimers2019sentence, xiong2020approximate}. Cross-encoder rerankers instead score a query and document jointly, usually improving precision at the cost of one expensive inference per candidate pair \citep{nogueira2019passage}. Hard-negative training improves bi-encoders, while listwise objectives improve rerankers \citep{xiong2020approximate, zhuang2022rankt5}; both rely on offline parameter updates. Our focus is complementary: retaining ranking feedback after deployment while leaving both the encoder and document index fixed.

\paragraph{Budgeted ranking feedback and inference.}
Resource constraints have been studied around retrieval, although not as persistent state in an embedding space. Active learning-to-rank selects which queries to annotate under a finite labeling budget \citep{bilgic2012active,wang2024simple}; online learning-to-rank updates a ranker from restricted top-$k$ feedback, and counterfactual methods learn from logged interactions even when queries do not repeat \citep{chaudhuri2016online,joachims2017unbiased}. Neural reranking work reduces the per-pair cost of cross-encoder feedback \citep{hui2022ed2lm}, while API routing allocates expensive model calls selectively across queries \citep{chen2023frugalgpt}. These lines of work motivate a workload-level feedback budget, but solve a different problem: they update a ranker or routing policy, or spend an expensive prediction on the current query. They do not retain relevance feedback as external state in the returned embedding space, compare global-wise, task-wise, and query-wise persistence, or preserve a potentially closed encoder and its existing index. We study this missing budgeted test-time regime for embedding retrieval.

\paragraph{Pseudo-relevance feedback and query adaptation.}
Pseudo-relevance feedback traditionally uses initially retrieved documents to refine a query. In neural dense retrieval, ANCE-PRF trains a separate query encoder that consumes the query and its top-ranked documents, sharing what is learned through the encoder parameters while keeping the document index fixed \citep{yu2021improving}. More closely related to our setting, TOUR directly optimizes an instance-level query representation from cross-encoder pseudo-labels \citep{sung2023optimizing}, and ReFIT distills a reranker's candidate scores into the test query vector before a second retrieval pass \citep{reddy2023refit}. A subsequent large-scale study evaluates this query-level distillation across domains, languages, modalities, update counts, and candidate depths \citep{gangi2025large}. GQR similarly refines a test-query embedding, using scores supplied by a complementary retriever in multimodal retrieval \citep{uzan2025guided}. These methods establish that external scores can improve a query representation at inference time, and some characterize how much compute to spend on that query, but their optimized state remains query-private. We instead study whether output-space feedback can persist across queries, at which scope it should persist, and how that decision changes under a global feedback budget.

\paragraph{Test-time training and inference-time compute.}
Test-time training adapts a deployed model using unlabeled test inputs, commonly by updating weights with a self-supervised objective \citep{sun2020test, liu2021ttt++}. A parallel line of work allocates additional inference-time computation to individual language-model inputs \citep{snell2024scaling}. Our setting differs in both the adaptation signal and the state being retained: the scarce resource is an external relevance judgment, and the adapted state is an additive vector outside the frozen embedding model. This makes test-time adaptation possible even with a closed-source model, provided that it returns embeddings and the user controls downstream vector search.

\section{Ranking Equivalence of Vector Addition}
\label{app:ranking_equivalence}

We formalize the ranking equivalence stated in the method section. The result concerns the exact score-induced
ordering; approximate search and candidate truncation are discussed
afterward.

\begin{proposition}[Additive-vector ranking equivalence]
Let $\xv,\vv\in\bR^p$, let $\alpha\in\bR$, and assume $\xv+\alpha\vv\neq\mathbf{0}$. Let every document be represented by a unit-normalized vector $\zv_d\in\bR^p$, and score documents by inner product with the normalized corrected query
\begin{equation}
    \widetilde{\xv}
    =
    \unitnorm{\xv+\alpha\vv}.
\end{equation}
Define the additive score
\begin{equation}
    S_{\mathrm{add}}(d)
    =
    \xv^{\top}\zv_d+\alpha\vv^{\top}\zv_d.
\end{equation}
Then, for any two documents $d_a$ and $d_b$,
\begin{equation}
    \widetilde{\xv}^{\top}\zv_{d_a}
    \geq
    \widetilde{\xv}^{\top}\zv_{d_b}
    \quad\Longleftrightarrow\quad
    S_{\mathrm{add}}(d_a) \geq S_{\mathrm{add}}(d_b).
    \label{eq:ranking_equivalence}
\end{equation}
Consequently, the two scoring rules induce the same pairwise order and the same ties.
\end{proposition}

\begin{proof}
Define $c=\norm{\xv+\alpha\vv}$. The nonzero-query assumption gives $c>0$. The score difference between any two documents is
\begin{align}
    \widetilde{\xv}^{\top}\zv_{d_a}
    -
    \widetilde{\xv}^{\top}\zv_{d_b}
    &=
    \frac{(\xv+\alpha\vv)^{\top}
    (\zv_{d_a}-\zv_{d_b})}{c} \\
    &=
    \frac{S_{\mathrm{add}}(d_a)-S_{\mathrm{add}}(d_b)}{c}.
\end{align}
Multiplication by the positive constant $1/c$ preserves whether this difference is positive, zero, or negative. This proves Equation~\eqref{eq:ranking_equivalence} and therefore preserves every pairwise ordering and tie.
\end{proof}

\paragraph{Scope and implementation limits.}
The proposition establishes equality of rankings, not equality of absolute scores: normalization rescales every score for a fixed query by $1/c$. For a shared scope $g$, the term $b_g(d)=\vv_g^{\top}\zv_d$ is independent of the query and may be cached as a per-document bias. However, adding this bias only to documents in the \emph{original} top-$k$ list need not recover the exact global top-$k$, because $\vv_g$ can promote a document that the base query did not retrieve. To target the global corrected ranking, one should instead issue a new search over the existing inner-product or cosine ANN index using $\widetilde{\xv}_q$ as the query vector. The indexed document vectors, vector dimensionality, and similarity function remain unchanged, although an approximate ANN algorithm may naturally return an approximation to the exact score-induced ranking.

\section{Scope-Balanced Query Weights}
\label{app:query_weights}

We now specify the weights in Equation~\eqref{eq:vector_objective}. Let $t(q)$ denote the task of query $q$, and define the feedback queries from task $t$ within scope $g$ as
\begin{equation}
    \cF_{g,t}
    =
    \{q\in\cF_g\mid t(q)=t\}.
\end{equation}
Let $\cT_g^+=\{t:\abs{\cF_{g,t}}>0\}$ be the tasks that contribute at least one feedback query to state $g$. We assign
\begin{equation}
    w_q
    =
    \frac{1}
    {\abs{\cT_g^+}\,\abs{\cF_{g,t(q)}}}.
    \label{eq:query_weight}
\end{equation}
These weights are normalized because
\begin{equation}
    \sum_{q\in\cF_g}w_q
    =
    \frac{1}{\abs{\cT_g^+}}
    \sum_{t\in\cT_g^+}
    \frac{1}{\abs{\cF_{g,t}}}
    \sum_{q\in\cF_{g,t}}1
    =1.
\end{equation}
Substituting Equation~\eqref{eq:query_weight} into Equation~\eqref{eq:vector_objective} yields the equivalent nested average
\begin{equation}
    \cL_g(\vv)
    =
    \frac{1}{\abs{\cT_g^+}}
    \sum_{t\in\cT_g^+}
    \frac{1}{\abs{\cF_{g,t}}}
    \sum_{q\in\cF_{g,t}}
    \kl\!\left(
        p_q^{\mathrm{T}}\,\middle\|\,p_q^{\mathrm{S}}(\vv)
    \right)
    +\lambda\norm{\vv}^2.
    \label{eq:nested_scope_objective}
\end{equation}
For a global-wise state, Equation~\eqref{eq:nested_scope_objective} first averages within each active task and then across tasks, matching evaluation that averages tasks equally and preventing a large task from dominating $\vv_g$. For a task-wise state, $\abs{\cT_g^+}=1$, so $w_q=1/\abs{\cF_g}$. For a query-wise state, $\abs{\cT_g^+}=\abs{\cF_g}=1$, so $w_q=1$. Finally, when every active task contributes the same number $n$ of queries, $\abs{\cF_g}=\abs{\cT_g^+}n$ and Equation~\eqref{eq:query_weight} also reduces to the flat query weight $w_q=1/\abs{\cF_g}$. The distinction matters only when tasks contribute different numbers of feedback queries.

\section{Bayesian Interpretation of Deployment Shrinkage}
\label{app:bayesian_shrinkage}

We generalize Proposition~\ref{prop:bayesian_shrinkage} to unequal prior and
per-query noise scales. This surrogate formalizes the uncertainty of a
vector estimate; it is not an additional assumption used to optimize
Equation~\eqref{eq:vector_objective}.

\begin{proposition}[Bayesian vector shrinkage]
Let $\uv_g\in\bR^p$ be a latent transferable vector and let $\boldsymbol{\Sigma}$ be positive definite. Suppose
\begin{equation}
    \uv_g \sim \mathcal{N}(\mathbf{0},\sigma_0^2\boldsymbol{\Sigma}),
    \qquad
    \vv_g^*\mid \uv_g,n_g
    \sim
    \mathcal{N}\!\left(\uv_g,\frac{\sigma^2}{n_g}\boldsymbol{\Sigma}\right),
    \label{eq:bayesian_vector_model}
\end{equation}
for $n_g\geq 1$. Then the posterior mean is
\begin{equation}
    \mathbb{E}[\uv_g\mid\vv_g^*,n_g]
    =
    \frac{n_g}{n_g+\kappa}\vv_g^*,
    \qquad
    \kappa=\frac{\sigma^2}{\sigma_0^2}.
    \label{eq:bayesian_vector_mean}
\end{equation}
In particular, equal prior and per-query noise scales give $\kappa=1$ and recover Equation~\eqref{eq:alpha}.
\end{proposition}

\begin{proof}
Gaussian conjugacy gives posterior precision
\begin{equation}
    \left(\frac{1}{\sigma_0^2}+\frac{n_g}{\sigma^2}\right)
    \boldsymbol{\Sigma}^{-1}.
\end{equation}
Multiplying its inverse by the likelihood precision times $\vv_g^*$ yields
\begin{align}
    \mathbb{E}[\uv_g\mid\vv_g^*,n_g]
    &=
    \left(\frac{1}{\sigma_0^2}+\frac{n_g}{\sigma^2}\right)^{-1}
    \frac{n_g}{\sigma^2}\vv_g^* \\
    &=
    \frac{n_g\sigma_0^2}{\sigma^2+n_g\sigma_0^2}\vv_g^*
    =
    \frac{n_g}{n_g+\sigma^2/\sigma_0^2}\vv_g^*.
\end{align}
\end{proof}

The $1/n_g$ variance in Equation~\eqref{eq:bayesian_vector_model} treats each feedback query as one independent unit of evidence about transfer. It therefore explains why $n_g$ counts queries rather than the number of documents scored within a query. Giving the zero-centered prior the precision of one feedback query sets $\kappa=1$---equivalently, it adds one zero-valued pseudo-query---and produces the parameter-free rule in Equation~\eqref{eq:alpha}. Defining $\alpha_g=0$ at $n_g=0$ recovers the no-adaptation state. The shared covariance shape in Equation~\eqref{eq:bayesian_vector_model} is what makes the posterior shrinkage scalar rather than dimension-specific.

\section{Additional Experimental Details}
\label{app:experimental_details}

\paragraph{Embedding and retrieval protocol.}
All query and document outputs are $\ell_2$ normalized, and retrieval uses inner
product, equivalently cosine similarity. We initialize
every $\vv_g$ at zero. After learning, evaluation issues a new search with the
adapted query against the complete, unchanged document index; it is not
restricted to reordering the rewarded candidates.

\paragraph{Retrieval tasks.}
The 15-task suite consists of ArguAna, CQADupstackGamingRetrieval,
CQADupstackUnixRetrieval, ClimateFEVERHardNegatives,
DBPediaHardNegatives, FEVERHardNegatives, FiQA2018,
HotpotQAHardNegatives, MSMARCOHardNegatives, NQHardNegatives,
QuoraRetrievalHardNegatives, SCIDOCS, TRECCOVID,
TopiOCQAHardNegatives, and Touche2020Retrieval.v3. We use the frozen test
query identifiers and task revisions recorded in the experiment manifest.
Across these tasks, the evaluation population contains $N=12{,}263$ queries
($43$--$1{,}595$ per task). The task-specific corpora contain
$1{,}661{,}208$ entries in aggregate and range from $8{,}674$ to $303{,}732$
entries.

\begin{table}[t]
    \centering
    \scriptsize
    \setlength{\tabcolsep}{4pt}
    \caption{Evaluation-set size for each MTEB retrieval task. ``HN''
    abbreviates the HardNegatives revision. The final row is the pooled count;
    reported performance nevertheless averages the 15 tasks equally.}
    \label{tab:task_statistics}
    \begin{tabular}{@{}lrr@{}}
        \toprule
        Task & Queries & Corpus entries \\
        \midrule
        ArguAna            & 1,406 &   8,674 \\
        CQA Gaming         & 1,595 &  45,301 \\
        CQA Unix           & 1,072 &  47,382 \\
        ClimateFEVER-HN    & 1,000 &  47,416 \\
        DBPedia-HN         &   400 &  90,070 \\
        FEVER-HN           & 1,000 & 163,698 \\
        FiQA2018           &   648 &  57,638 \\
        HotpotQA-HN        & 1,000 & 225,621 \\
        MSMARCO-HN         &    43 &   8,812 \\
        NQ-HN              & 1,000 & 198,779 \\
        Quora-HN           & 1,000 & 177,163 \\
        SCIDOCS            & 1,000 &  25,657 \\
        TRECCOVID          &    50 & 171,332 \\
        TopiOCQA-HN        & 1,000 &  89,933 \\
        Touche2020 v3      &    49 & 303,732 \\
        \midrule
        Total              & 12,263 & 1,661,208 \\
        \bottomrule
    \end{tabular}
\end{table}

\paragraph{Reward-model scoring interfaces.}
\label{app:reward_interfaces}

Every reward model scores the same frozen query--document pairs independently.
All models receive the same generic retrieval instruction,
``Given a web search query, retrieve relevant passages that answer the
query,''. 

For Qwen3-Reranker-4B, the system message asks whether the document meets the
requirements of the query and instruction and restricts the response to
\emph{yes} or \emph{no}. Following the model's native reranking template, we
append an empty thinking block and take the raw next-token logits without
sampling text. Its reward is the normalized \emph{Yes} probability in
Equation~\eqref{eq:reranker_reward}. Inputs are truncated to the model's
8,192-token limit.

The four general-purpose LLMs receive the following shared relevance rubric:
$5$, perfectly relevant and directly and completely satisfies the query;
$4$, highly relevant but incomplete in detail or coverage; $3$, partially
relevant, providing a partial or indirect answer that requires inference;
$2$, marginally relevant and on topic but not answering the query; $1$,
barely related through background or terminology; and $0$, unrelated. We
convert the decoded final grade to
\begin{equation}
    r(q,d)=\frac{s(q,d)}{5}, \qquad s(q,d)\in\{0,1,\ldots,5\}.
    \label{eq:rubric_reward}
\end{equation}
Qwen3-8B-Instruct and Qwen3-VL-30B-A3B-Instruct run with model thinking
disabled and are instructed to reason briefly before ending with exactly
\texttt{Score: X}. We decode that terminal digit from one sample generated
with the official non-thinking configuration (temperature $0.7$, top-$p$
$0.8$, top-$k$ $20$), a 512-token limit, and seed 1234. Gemini 3.1 Pro is queried with
\texttt{thinkingLevel=low}; its internal thought content is discarded and
only the decoded final digit is retained. Because
Table~\ref{tab:reward_model_quality} uses each reward model only to order its
own scored candidate set, dividing the grade by five does not change the
reported ranking.

\paragraph{Candidate and budget construction.}
The allocation rule is an evaluation protocol, not a component of
\textsc{TTT-Embed}; the method accepts any externally selected reward set.
We use a task-balanced water-filling schedule. Within task $t$, queries follow
a permutation generated with seed $1000+t$, where $t$ is the task's
position in the fixed 15-task list. Water filling cycles through the active
tasks and adds one query from each task in turn; exhausted tasks drop out.
Within a query, candidates follow its own base-retrieval order. This produces
nested prefixes, so increasing a budget never removes a previously rewarded
query--document pair.

The main depth-10 breadth grid uses
\begin{equation}
\begin{aligned}
    m\in\{&1,2,4,8,15,30,60,120,\\
           &240,480,960,1920,3840,7680,12263\},
\end{aligned}
\end{equation}
where $m$ is the total number of rewarded queries and each contributes its
first ten candidates, hence $B=10m$ and $b=10m/N$. After all queries reach
depth 10, the curve adds full-coverage depths 30 and 100. Consequently, all
scopes and baselines see exactly the same prefix of query--document pairs at a
given budget. Every selected candidate has a real reward-model score; no
missing value is masked or imputed.

For the confidence factor, $n_g$ is the number of distinct rewarded queries
used to learn a vector: $n_g=m$ for the global vector, the task-specific
rewarded-query count for each task vector, and $n_g=1$ for every query-wise
vector. A scope with no rewarded query has no learned vector and applies zero
correction. For direct reranking, rewarded candidates are stably sorted by
reward within their original slots; unrewarded slots and queries remain in
base-retrieval order. The main curves and objective ablation evaluate the
fixed population of all $N=12{,}263$ queries. The held-out transfer
experiments instead use disjoint reward and evaluation queries.

For readability, Table~\ref{tab:scope_peak} denotes the low-budget
$m=15$, depth-10 operating point as $b=10^{-2}$; its exact normalized budget
on the frozen grid is $b=(15\times10)/12{,}263=0.0122$.

\paragraph{Area under the budget curve.}
For a matched gain curve $\Delta(b)$, we define
\begin{equation}
    \operatorname{AUBC}_{\Delta}[b_{\min},b_{\max}]
    =
    \frac{
        \int_{\log b_{\min}}^{\log b_{\max}}
        \Delta(e^u)\,\mathrm{d}u
    }{
        \log b_{\max}-\log b_{\min}
    }.
    \label{eq:aubc}
\end{equation}
We use trapezoidal integration on the frozen log-budget grid. The zero-budget
base point is excluded because $\log 0$ is undefined. AUBC@100 uses
$b_{\min}=10/N$ and $b_{\max}=100$; AUBC@8 uses the corresponding common grid
up to $b=8$.

\section{Full-Curve Budget Results}
\label{app:aubc_results}

Table~\ref{tab:aubc_full} reports the model-level summary underlying
Figure~\ref{fig:budget_scope}. Although query-private optimization attains the
best high-budget endpoint, task-wise sharing is the most robust choice when
performance is integrated over the full log-budget range.

\begin{table}[h]
    \centering
    \small
    \caption{Gain AUBC@100 ($\times100$) for the complete budget curves.
    Higher is better. Each value integrates the matched
    $\Delta\mathrm{nDCG@10}$ curve uniformly in log budget.}
    \label{tab:aubc_full}
    \resizebox{\columnwidth}{!}{%
    \begin{tabular}{lrrrr}
        \toprule
        Embedding model & Global-wise & Task-wise & Query-wise & Direct rerank \\
        \midrule
        Qwen3-Embedding-0.6B & 3.65 & \textbf{4.17} & 2.98 & 2.28 \\
        Qwen3-Embedding-4B   & 3.85 & \textbf{4.10} & 1.70 & 1.18 \\
        EmbeddingGemma-300M  & 3.84 & \textbf{5.45} & 3.42 & 2.70 \\
        Gemini Embedding 1   & 1.53 & \textbf{3.09} & 2.01 & 1.34 \\
        Gemini Embedding 2   & 0.48 & \textbf{1.69} & 1.15 & 0.71 \\
        \bottomrule
    \end{tabular}
    }
\end{table}

\section{Full Transfer Curves}
\label{app:transfer_curves}

\begin{figure*}[t]
    \centering
    \includegraphics[width=0.95\textwidth]{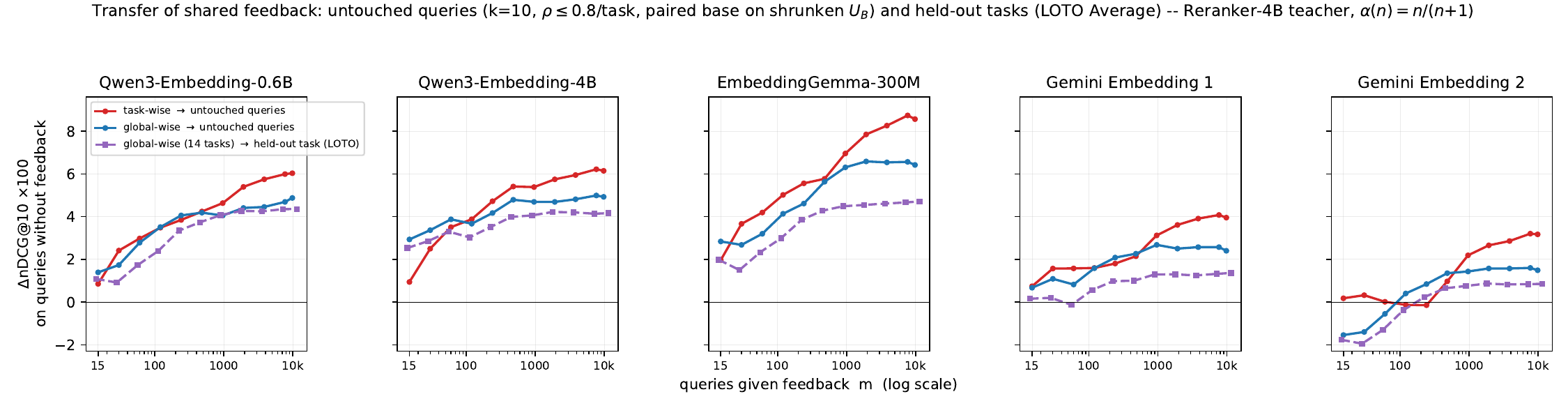}
    \caption{Transfer of shared feedback for each embedding model. Red and
    blue curves evaluate task-wise and global-wise states only on held-out
    queries, using a paired raw-query baseline on the same shrinking set.
    Purple curves learn a global-wise $\vv_g$ from 14 tasks and evaluate it on
    the held-out fifteenth task. All curves use Reranker-4B feedback and the fixed
    dose $\alpha(n)=n/(n+1)$.}
    \label{fig:transfer_full}
\end{figure*}

\section{Same-Budget Specialization Recovery Results}
\label{app:sft_recovery}

Table~\ref{tab:sft_recovery_full} separates the specialization tradeoff from
the subsequent recovery experiment. The recovery states in panel (b) are
learned and activated only on the 15 MTEB tasks; the SKILLRET score in panel
(a) therefore remains the specialized model's raw retrieval score. If a
state were learned on SKILLRET alone, global-wise and task-wise sharing would
coincide because SKILLRET contains a single retrieval task. They differ in
panel (b) because global-wise recovery pools feedback across all 15 MTEB
tasks, whereas task-wise recovery learns one state per MTEB task.

\begin{table}[h]
    \centering
    \small
    \setlength{\tabcolsep}{5pt}
    \caption{\textbf{Complete matched-budget recovery comparison.}
    Scores are nDCG@10 $\times100$. Panel (a) evaluates embedding-only
    retrieval before test-time feedback. Panel (b) freezes
    SKILLRET-SFT-0.6B and gives every method the same $b=10$ Reranker-4B
    feedback pairs; gains are relative to its 45.42 MTEB score.}
    \label{tab:sft_recovery_full}

    \resizebox{\columnwidth}{!}{%
    \begin{tabular}{@{}lrr@{}}
        \toprule
        \multicolumn{3}{@{}l}{\textit{(a) Specialization tradeoff}} \\
        Model & SKILLRET & MTEB Average \\
        \midrule
        Qwen3-Embedding-0.6B & 52.24 & 49.90 \\
        SKILLRET-SFT-0.6B    & 66.21 & 45.42 \\
        \addlinespace[4pt]
        \multicolumn{3}{@{}l}{\textit{(b) MTEB recovery on
        SKILLRET-SFT-0.6B}} \\
        Deployment state & MTEB Average & Gain \\
        \midrule
        Global-wise \textsc{TTT-Embed} & 50.38 & +4.97 \\
        Task-wise \textsc{TTT-Embed}   & 53.42 & +8.00 \\
        Query-wise \textsc{TTT-Embed}  & 53.83 & +8.42 \\
        Direct rerank                  & 51.58 & +6.16 \\
        \bottomrule
    \end{tabular}%
    }
\end{table}

\section{Connections Between Distillation, Reverse KL, and Softmax-GRPO}
\label{app:objective_analysis}

We analyze the data-fitting terms below; all three reported objectives add
the same ridge penalty on $\vv$. For one query, omit the query index and let
\begin{equation}
    p_i
    =
    \frac{\exp(r_i/\tau_{\mathrm T})}
         {\sum_{j=1}^{K}\exp(r_j/\tau_{\mathrm T})},
    \qquad
    \pi_i
    =
    \frac{\exp(\ell_i)}
         {\sum_{j=1}^{K}\exp(\ell_j)}
    \label{eq:app_teacher_student}
\end{equation}
denote the reward-model teacher distribution and the policy induced by the
corrected query embedding, respectively. Here $\ell_i$ includes the
student-side temperature, and $K=10$ in the objective ablation.

\paragraph{Knowledge distillation.}
The objective used by \textsc{TTT-Embed} is
\begin{align}
    \mathcal L_{\mathrm{KD}}
    &=
    \kl(p\mid\mid\pi) \\
    &=
    -\sum_{i=1}^{K}p_i\log\pi_i-H(p).
    \label{eq:app_kd}
\end{align}
Because the teacher entropy $H(p)$ is constant, this is precisely
soft-target knowledge distillation \citep{hinton2015distilling}, instantiated
listwise over retrieved documents as in listwise learning to rank
\citep{cao2007learning}. Its gradient with respect to each student logit is
\begin{equation}
    \frac{\partial\mathcal L_{\mathrm{KD}}}{\partial \ell_i}
    =
    \pi_i-p_i.
    \label{eq:app_kd_gradient}
\end{equation}
Thus every scored document contributes a deterministic, bounded correction,
and the teacher distribution is the exact target whenever it is representable
by the $\vv$ parameterization.

\paragraph{Reverse KL as entropy-regularized reward maximization.}
Substituting $p_i=\exp(r_i/\tau_{\mathrm T})/Z_{\mathrm T}$ gives
\begin{align}
    \mathcal L_{\mathrm{RKL}}
    &=
    \kl(\pi\mid\mid p) \\
    &=
    -H(\pi)
    -\frac{1}{\tau_{\mathrm T}}\bbE_{i\sim\pi}[r_i]
    +\log Z_{\mathrm T}.
    \label{eq:app_rkl_reward}
\end{align}
Minimizing reverse KL is therefore exactly equivalent to maximizing
$\bbE_{\pi}[r_i]+\tau_{\mathrm T}H(\pi)$. In this sense, reverse KL is
closer to entropy-regularized policy optimization than to teacher-weighted
distillation. Its logit gradient,
\begin{equation}
    \frac{\partial\mathcal L_{\mathrm{RKL}}}{\partial \ell_i}
    =
    \pi_i
    \left[
        \log\frac{\pi_i}{p_i}
        -\kl(\pi\mid\mid p)
    \right],
    \label{eq:app_rkl_gradient}
\end{equation}
weights discrepancies by the student's current probability. It can therefore
concentrate updates on modes already favored by the student, whereas
Equation~\eqref{eq:app_kd_gradient} directly corrects every teacher--student
probability mismatch.

\paragraph{Connection to Softmax-GRPO.}
Our budget-controlled adaptation of Retrieval-GRPO
\citep{liu2025taosearchemb} treats the $K$ documents as actions. Let
$\pi_i^0$ be the policy induced by the unmodified query, define
\begin{equation}
    A_i=\frac{r_i-\overline r}{\sigma_r+\delta},
    \qquad
    \rho_i=\frac{\pi_i}{\pi_i^0},
\end{equation}
and note that $\sum_i A_i=0$. The clipped document-level objective is
\begin{equation}
    \mathcal L_{\mathrm{SG}}
    =
    -\frac{1}{K}\sum_{i=1}^{K}
    \min\!\left\{
        \rho_i A_i,\,
        \operatorname{clip}(\rho_i,1-\epsilon,1+\epsilon)A_i
    \right\}.
    \label{eq:app_softmax_grpo}
\end{equation}
Unlike offline Retrieval-GRPO, we freeze the encoder and candidate set,
optimize only $\vv$, and use the native query policy as the fixed reference.
At initialization, $\pi=\pi^0$, the ratio is one and clipping is inactive.
Differentiating Equation~\eqref{eq:app_softmax_grpo} yields
\begin{equation}
    \left.
    \frac{\partial\mathcal L_{\mathrm{SG}}}{\partial \ell_i}
    \right|_{\pi=\pi^0}
    =
    -\frac{A_i}{K}.
    \label{eq:app_softmax_grpo_gradient}
\end{equation}
This exposes a local connection to knowledge distillation. If the reference
policy is approximately uniform and the reward spread is small relative to
$\tau_{\mathrm T}$, then
\begin{equation}
    p_i
    =
    \frac{1}{K}
    +
    \frac{r_i-\overline r}{K\tau_{\mathrm T}}
    +
    O\!\left(
        \frac{\norm{r-\overline r\mathbf 1}^{2}}{\tau_{\mathrm T}^{2}}
    \right).
    \label{eq:app_teacher_local}
\end{equation}
Equations~\eqref{eq:app_kd_gradient} and
\eqref{eq:app_softmax_grpo_gradient} are then proportional to the same
centered-reward direction. Away from this local regime, however, knowledge
distillation retains the calibrated teacher probabilities and has $p$ as its
explicit target. Softmax-GRPO standardizes away reward magnitude and performs
a clipped step relative to $\pi^0$; it need not converge to $p$. This
distinction explains how both objectives can learn useful directions while
knowledge distillation remains more consistent when one direction is pooled
across queries.

\section{Full Optimization-Objective Ablation}
\label{app:loss_full}

Table~\ref{tab:loss_full} reports the average across five models at all four
tested feedback breadths. Knowledge distillation has the strongest global-wise
average at every budget and the strongest task-wise average at
all three points with $b\geq0.196$; reverse KL is slightly better query-wise
in that regime. At the underlying model level, knowledge distillation wins 40
of 60 comparisons and all 15 corresponding task-wise cells. Reverse KL wins
16 cells and Softmax-GRPO wins four. All candidates in this ablation have
real reward-model scores; this complete support avoids conflating an
objective's behavior with an arbitrary score assigned to missing candidates.

\begin{table}[h]
    \centering
    \small
    \setlength{\tabcolsep}{3.5pt}
    \caption{Complete objective ablation averaged over five embedding models
    ($\Delta\mathrm{nDCG@10}\times100$). Each selected query has ten judged
    candidates. With $m$ feedback queries among $N=12{,}263$ total queries,
    the normalized budget is $b=10m/N$.}
    \label{tab:loss_full}
    \begin{tabular}{@{}rlrrr@{}}
        \toprule
        $b$ & Objective
            & \shortstack{Global-\\wise}
            & \shortstack{Task-\\wise}
            & \shortstack{Query-\\wise} \\
        \midrule
        0.012 & Knowledge Distillation & \textbf{1.24} & 0.93 & \textbf{0.006} \\
              & Softmax-GRPO           & 0.78 & 0.28 & 0.005 \\
              & Reverse KL             & $-2.47$ & \textbf{1.13} & 0.002 \\
        \addlinespace
        0.196 & Knowledge Distillation & \textbf{3.27} & \textbf{3.40} & 0.64 \\
              & Softmax-GRPO           & 1.76 & 1.84 & 0.41 \\
              & Reverse KL             & 2.10 & $-0.92$ & \textbf{0.68} \\
        \addlinespace
        3.13 & Knowledge Distillation & \textbf{3.99} & \textbf{6.07} & 2.92 \\
             & Softmax-GRPO           & 2.07 & 3.03 & 1.94 \\
             & Reverse KL             & 3.71 & 4.17 & \textbf{2.99} \\
        \addlinespace
        10 & Knowledge Distillation & \textbf{4.00} & \textbf{6.38} & 5.26 \\
           & Softmax-GRPO           & 2.07 & 3.13 & 3.27 \\
           & Reverse KL             & 3.72 & 4.66 & \textbf{5.30} \\
        \bottomrule
    \end{tabular}
\end{table}

\subsection{Optimization Details} 
We zero-initialize each $\vv_g$ and optimize
Equation~\eqref{eq:vector_objective} for 300 full-batch Adam steps
\citep{kingma2014adam}, using learning rate $0.01$, $\lambda=0.3$,
$\tau_{\mathrm{T}}=0.02$, and $\tau_{\mathrm{S}}=0.05$. These settings are
fixed across embedding models, tasks, scopes, and budgets. Reverse KL uses
the same teacher and student temperatures. Softmax-GRPO uses policy
temperature $0.05$, the policy at $\vv=\mathbf{0}$ as its fixed reference,
ratio clipping $\epsilon=0.2$, and an advantage-standard-deviation floor of
$10^{-4}$. All objectives use the same rewarded candidates and the same
300-step optimizer budget.

\end{document}